\documentclass[copyright,creativecommons]{eptcs}
\providecommand{\event}{EXPRESS/SOS 2015}

\usepackage{amsmath}
\usepackage{amsthm}
\usepackage{amssymb}
\usepackage{verbatim} 
\usepackage{stmaryrd}
\usepackage{paralist}
\usepackage{cite}
\usepackage{appendix}
\usepackage{hyperref}
\usepackage{subfig}
\usepackage{array}
\usepackage{textcomp}
\usepackage{txfonts}
\usepackage{latexsym}
\usepackage{yfonts}
\usepackage{empheq}

\allowdisplaybreaks

\newif\ifdraft\draftfalse

\newcommand{\bisim}{\sim}
\newcommand{\bbisim}{\sim_b}
\newcommand{\cbisim}{\sim_c}
\newcommand{\abisim}{\sim_a}

\newcommand{\relR}{\mathrel{\textsf{R}}}

\newcommand{\closed}[1]{{#1}\text{-closed}}

\newcommand{\T}{\textsf{T}}
\newcommand{\N}{\mathbb{N}}

\newcommand{\Q}{\mathbb{Q}}

\newcommand{\ddedrule}[2]{\frac{\displaystyle #1}{\displaystyle #2}}
\newcommand{\dedrule}[2]{\frac{#1}{#2}}
\newcommand{\trans}[1][]{\xrightarrow{\, {#1} \, }}
\newcommand{\ntrans}[1][]{\mathrel{{\trans[#1]}\makebox[0em][r]{$\not$\hspace{2ex}}}{\!}}

\newcommand{\otrans}[1][]{\overset{#1}{\rightsquigarrow}}

\newcommand{\TS}{\ensuremath{s}}
\newcommand{\DTS}{\ensuremath{d}}

\newcommand{\Ssignature}{\Sigma_\TS}
\newcommand{\Dsignature}{\Sigma_\DTS}

\newcommand{\te}{{\textswab{t}}}

\newcommand{\gtgeq}{\trianglerighteq}

\newcommand{\strans}[1][]{\stackrel{#1}{\longrightarrow}}

\newcommand{\tuple}[1]{\langle{#1}\rangle}

\newcommand{\rank}{\mathop{\sf rk}}
\newcommand{\arity}{\mathop{\sf ar}}

\newcommand{\openT}{\mathbb{T}}

\newcommand{\openTerms}{\openT(\Sigma)}
\newcommand{\openSTerms}{\openT(\Sigma_{\TS})}
\newcommand{\closedSTerms}{\T(\Sigma_{\TS})}

\newcommand{\openDTerms}{\openT(\Sigma_{\DTS})}
\newcommand{\closedDTerms}{\T(\Sigma_{\DTS})}
\newcommand{\closedTerms}{\T(\Sigma)}

\newcommand{\pprem}[1]{\textrm{pprem}(#1)}
\newcommand{\nprem}[1]{\textrm{nprem}(#1)}
\newcommand{\qprem}[1]{\textrm{qprem}(#1)}
\newcommand{\prem}[1]{\textrm{prem}(#1)}
\newcommand{\conc}[1]{\textrm{conc}(#1)}

\newcommand{\Red}[2]{\textrm{Red}}

\newcommand{\DVar}{\mathcal{V}\!_d}
\newcommand{\Var}{\mathcal{V}}

\newcommand{\Tr}{\textsf{Tr}}
\newcommand{\CerTr}{{\sf CT}}
\newcommand{\PosTr}{{\sf PT}}

\newcommand{\support}{\mathsf{supp}}

\newcommand{\proj}{\mbox{\small$\mathsf{\Pi}$}}

\newcommand{\nullproc}{0}

\newcommand{\Act}{A}

\newcommand{\ntmuft}{\ensuremath{\mathit{nt}\mkern-1.25mu \mu\mkern-1.0mu \mathit{f}\mkern-1.5mu \theta}}

\newcommand{\ntmuxt}{\ensuremath{\mathit{nt}\mkern-1.5mu \mu\mkern-0.75mu \mathit{x}\mkern-0.5mu \theta}}

\newcommand{\ntmufxt}{\ensuremath{\mathit{nt}\mkern-1.25mu \mu\mkern-1.0mu \mathit{f}\mkern-1.5mu \theta\mkern-1.75mu /\mkern-1.75mu \mathit{nt}\mkern-1.5mu \mu\mkern-0.75mu \mathit{x}\mkern-0.5mu \theta}}

\newcommand{\Diag}{\mathop{\textsf{Diag}}}
\newcommand{\diag}[1]{\textsf{Diag}\{#1\}}

\newcommand{\limp}{\Rightarrow}

\newcommand{\dist}[1]{\boldsymbol{#1}}

\newcommand{\sem}[2]{\llbracket{#1}\rrbracket_{#2}}

\newcommand{\TVar}{\mathcal{V}}

\newcommand{\proves}{\vdash}

\newcommand{\tm}{{\tilde{m}}}
\newcommand{\tM}{{\tilde{M}}}

\newcommand{\logic}{\mathcal{L}}
\newcommand{\blogic}{\logic_{b}}
\newcommand{\clogic}{\logic_{c}}
\newcommand{\alogic}{\logic_{a}}
\newcommand{\ologic}{\logic_{o}}
\newcommand{\xlogic}{\logic_{\chi}}

\newcommand{\pand}{\bigsqcap}

\newcommand{\ctop}{\text{(i)}}
\newcommand{\cland}{\text{(iv)}}
\newcommand{\cneg}{\text{(v)}}
\newcommand{\ca}{\text{(ii)}}
\newcommand{\cac}{\text{(iii)}}
\newcommand{\cmeasure}{\text{(vi)}}
\newcommand{\cpand}{\text{(vii)}}

\newcommand{\dia}[2]{\langle #1 \rangle {#2}}
\newcommand{\diac}[2]{\langle #1 \rangle_c {#2}}

\newcommand{\remarkPRD}[1]{}
\newcommand{\remarkMDL}[1]{}

\newcommand{\hrmkPRD}[1]{}

\newtheorem{definition}{Definition}
\newtheorem{theorem}{Theorem}
\newtheorem*{theorem*}{Theorem}
\newtheorem{lemma}{Lemma}
\newtheorem*{lemma*}{Lemma}

\title{SOS rule formats for \\ convex and abstract probabilistic bisimulations%
  \thanks{
    Supported by ANPCyT project PICT-2012-1823, SeCyT-UNC
    program 05/BP12 and their related projects, and EU 7FP
    grant agreement 295261 (MEALS).}%
}

\author{Pedro R. D'Argenio \qquad\qquad Matias David Lee
\institute{
  FaMAF, Universidad Nacional de C\'ordoba -- CONICET. \\
  Ciudad Universitaria, X5000HUA -- C\'ordoba, Argentina.}
\email{dargenio@famaf.unc.edu.ar \quad\qquad lee@famaf.unc.edu.ar}
\and
Daniel Gebler
\institute{
        Department of Computer Science, VU University Amsterdam,\\
        De Boelelaan 1081a, NL-1081~HV~Amsterdam, The Netherlands}
\email{e.d.gebler@vu.nl}
}

\def\titlerunning{SOS rule formats for convex and abstract probabilistic bisimulations}
\def\authorrunning{D'Argenio, Lee \& Gebler}

\begin{document}
 
\maketitle

\begin{abstract}
  Probabilistic transition system specifications (PTSSs) in the
  \ntmufxt\ format provide structural operational semantics for
  Segala-type systems that exhibit both probabilistic and
  nondeterministic behavior and guarantee that bisimilarity is a
  congruence for all operator defined in such format.
  Starting from the \ntmufxt, we obtain restricted formats that
  guarantee that three coarser bisimulation equivalences are
  congruences.
  We focus on
  \begin{inparaenum}[(i)]
  \item%
    Segala's variant of bisimulation that considers combined
    transitions, which we call here \emph{convex bisimulation};
  \item%
    the bisimulation equivalence resulting from considering Park \&
    Milner's bisimulation on the usual stripped probabilistic
    transition system (translated into a labelled transition system),
    which we call here \emph{probability obliterated bisimulation};
    and
  \item%
    a \emph{probability abstracted bisimulation}, which, like
    bisimulation, preserves the structure of the distributions but
    instead, it ignores the probability values.
  \end{inparaenum}
  In addition, we compare these bisimulation equivalences and provide a
  logic characterization for each of them.
\end{abstract}


\section{Introduction}

Structural operational semantics (SOS for short)~\cite{Plo04} is a
powerful tool to provide semantics to programming languages.  In SOS,
process behavior is described using transition systems and the
behavior of a composite process is given in terms of the behavior of
its components.
SOS has been formalized using an algebraic framework as
\emph{Transition Systems Specifications
  (TSS)}~\cite[etc.]{GV92,BIM95,Gro93,BG96,MRG07}.
Basically, a TSS contains a signature, a set of actions or labels, and
a set of rules. The signature defines the terms in the
language. The set of actions represents all possible
activities that a process (i.e., a term over the signature) can
perform. The rules define how a process should behave (i.e., perform
certain activities) in terms of the behavior of its subprocesses, that
is, the rules define compositionally the transition system associated
to each term of the language.
A particular focus of these formalizations was to provide a
meta-theory that ensures a diversity of semantic properties by 
simple inspection on the form of the rules.
(See~\cite{AFV01,MRG07,AFV01:curretn-trends} for overviews.)
One of such kind of properties is to ensure that a given equivalence
relation is a congruence for all operators whose semantics is defined
in a TSS whose rules complies to a particular format.  These so called
\emph{congruence theorems} have been proved for a variety of
equivalences in the non-probabilistic
case~\cite[etc.]{GV92,BIM95,Gro93}.

The introduction of probabilistic process algebras motivated the need
for a theory of structural operational semantics to define
\emph{probabilistic} transition systems.
Few earlier results appeared in this
direction~\cite{Bar02,Bar04,LT05,LT09} presenting congruence theorems
for Larsen \& Skou bisimulation equivalence~\cite{LS91}.
Most of these formats have complicated restrictions that extend to sets
of rules due to the fact that they considered transitions labeled
both with an action and a probability value.
By using a more modern view of probabilistic transition systems (where
the target of the transition is a probability distribution on states)
we manage to obtained the most general format for bisimulation
equivalence, which we called \ntmufxt, following the nomenclature
of~\cite{GV92,Gro93}.

Starting from the \ntmufxt\ format, in this paper we define formats to
guarantee that three coarser versions of bisimulation equivalence are
congruences for all operator definable in the respective format.
The first relation we focus on is Segala's variant of bisimulation
that considers combined transitions, here called \emph{convex
  bisimulation}~\cite{Seg95a}.
The second relation we explore originates here and we call it
\emph{probability abstracted bisimulation}.  Like bisimulation and
unlike convex bisimulation, it preserves the structure of the
distributions of each transition, but instead, it ignores the
probability values.  This relation preserves the fairness introduced by
the probability distributions.
Finally, we study the bisimulation equivalence resulting from
considering Park \& Milner's bisimulation \cite{Mil89} on the usual
stripped probabilistic transition system (translated into a labeled
transition system).  Here we call it \emph{probability obliterated
  bisimulation}.  This is the usual way to abstract probabilities, but
it has the drawback that it breaks the basic fairness provided by
probabilistic choices.

Apart from presenting congruence theorems for all previously mentioned
bisimulation equivalences, we briefly study alternative definitions of
these bisimulations, compare them with each other, and provide
logical characterizations, which are particularly new here for
probability abstracted and probability obliterated bisimulation
equivalences.

The paper is organized as follows. 
Sec.~\ref{sec:preliminaries} recalls the type of algebraic structure
and Sec.~\ref{sec:ptss} provides the basic notions and results of
probabilistic transition system specifications (PTSS).
Sec.~\ref{sec:semantic_relations} presents the different bisimulation
equivalences and a brief study of them, including their logical
characterizations.
The study of all the PTSS formats and the respective congruence
theorems is given in Sec.~\ref{sec:formats}.
The paper concludes in Sec~\ref{sec:conclusion}.


\section{Preliminaries}\label{sec:preliminaries}%

Let $S=\{\TS,\DTS\}$ be a set denoting two sorts.
Elements of sort $\TS \in S$ are intended to represent states in the
transition system, while elements of sort $\DTS \in S$ will represent
distributions over states.
We let $\sigma$ range over $S$. 
An \emph{$S$-sorted signature} is a structure $(F, \arity)$, where 
\begin{inparaenum}[(i)]
	\item $F$ is a set of \emph{function names}, and
	\item $\arity \colon F \to (S^* \times S)$ is the \emph{arity function}.
\end{inparaenum}
The rank of $f \in F$ is the number of arguments of $f$, defined by
$\rank(f) = n$ if $\arity(f)= \sigma_1 \ldots \sigma_n \to \sigma$.
(We write ``$\sigma_1 \ldots \sigma_n \to \sigma$'' instead of
``$(\sigma_1 \ldots \sigma_n,\sigma)$'' to highlight that function $f$ maps to sort $\sigma$.)  Function $f$ is
a \emph{constant} if $\rank(f) = 0$. To simplify the presentation we
will write an $S$-sorted signature $(F, \arity)$ as a pair of disjoint
signatures $(\Ssignature,\Dsignature)$ where $\Ssignature$ is the set
of operations that map to $\TS$ and $\Dsignature$ is the set of
operations that map to $\DTS$.
Let $\TVar$ and $\DVar$ be two infinite sets of $S$-sorted variables where $\TVar,\DVar,F$ are all mutually disjoint. We use $x,y,z$ (with possible sub- or super-scripts) to range over $\TVar$, $\mu,\nu$ to range over $\DVar$ and $\zeta$ to range over $\TVar \cup \DVar$. 

\begin{definition} \label{def:state_distribution_terms}
  Let $\Ssignature$ and $\Dsignature$ be two signatures as before and
  let $V \subseteq \TVar$ and $D \subseteq \DVar$.  We simultaneously
  define the sets of \emph{state terms} $T(\Ssignature,V,D)$
  and \emph{distribution terms} $T(\Dsignature,V,D)$ as the smallest
  sets satisfying:
\begin{inparaenum}[(i)]
  \item%
    $V \subseteq T(\Ssignature,V,D)$;
  \item%
    $D \subseteq T(\Dsignature,V,D)$;
  \item%
    $f(\xi_1, \cdots, \xi_{\rank(f)}) \in T(\Sigma_\sigma,V,D)$, if
    $\arity(f) = \sigma_1\ldots\sigma_n \to \sigma$ and
    $\xi_i \in T(\Sigma_{\sigma_i},V,D)$.
\end{inparaenum}
\end{definition}

We let $\openTerms = T(\Ssignature,\TVar,\DVar) \cup
T(\Dsignature,\TVar,\DVar)$ denote the set of all \emph{open terms}
and distinguish the sets $\openSTerms = T(\Ssignature,\TVar,\DVar)$ of
\emph{open state terms} and $\openDTerms = T(\Dsignature,\TVar,\DVar)$
of \emph{open distribution terms}.
Similarly, we let $\closedTerms = T(\Ssignature,\emptyset,\emptyset)
\cup T(\Dsignature,\emptyset,\emptyset)$ denote the set of all
\emph{closed terms} and distinguish the sets $\closedSTerms =
T(\Ssignature,\emptyset,\emptyset)$ of \emph{closed state terms} and
$\closedDTerms = T(\Dsignature,\emptyset,\emptyset)$ of \emph{closed
  distribution terms}.
We let $t$, $t'$, $t_1$,\ldots\ range over state terms, $\theta$,
$\theta'$, $\theta_1$,\ldots\ range over distribution terms, and $\xi$,
$\xi'$, $\xi_1$,\ldots\ range over any kind of terms.
With $\Var(\xi) \subseteq \TVar \cup \DVar$ we denote the set of variables
occurring in term $\xi$.

Let $\Delta(\closedSTerms)$ denote the set of all (discrete) probability distributions on $\closedSTerms$. 
We let $\pi$ range over $\Delta(\closedSTerms)$.
For each $t \in \closedSTerms$, let $\delta_t \in \Delta(\closedSTerms)$ denote the \emph{Dirac distribution}, i.e., $\delta_t(t)=1$ and  $\delta_t(t')=0$ if $t$ and $t'$ are not syntactically equal.  
For $X \subseteq \closedSTerms$ we define $\pi(X)=\sum_{t\in X}\pi(t)$. The convex combination $\sum_{i \in I}p_i \pi_i$ of a family $\{\pi_i\}_{i \in I}$ of probability distributions with $p_i \in (0,1]$ and $\sum_{i \in I} p_i = 1$ is defined by $(\sum_{i \in I}p_i \pi_i)(t) = \sum_{i \in I} (p_i \pi_i(t))$. 

The type of signatures we consider has a particular construction.  We
start from a signature $\Ssignature$ of functions mapping into sort
$\TS$ and construct the signature $\Dsignature$ of functions mapping
into $\DTS$ as follows.
For each $f \in F_\TS$ we include a function symbol
$\dist{f} \in F_\DTS$ with $\arity(\dist{f}) = \DTS\ldots\DTS\to\DTS$
and $\rank(\dist{f}) = \rank(f)$.  We call $\dist{f}$
the \emph{probabilistic lifting} of $f$.  (We use boldface fonts to
indicate that a function in $\Dsignature$ is the probabilistic lifting
of another in $\Ssignature$.)
Moreover $\Dsignature$ may include any of the following additional
operators:
\begin{inparaenum}[(i)]
\item%
  $\delta$ with arity $\arity(\delta) = \TS\to\DTS$, and
\item%
  $\bigoplus_{i\in I}[p_i]\textvisiblespace$ with
  $I$ being a finite or countable infinite index set,
  $\sum_{i \in I} p_i = 1$,
  $p_i\in(0,1]$ for all $i\in I$, and
  $\arity\left(\bigoplus_{i\in I}[p_i]\textvisiblespace\right) =
  \DTS^{|I|} \to \DTS$.
\end{inparaenum}
Notice that if $I$ is countably infinite,
$\bigoplus_{i\in I}[p_i]\textvisiblespace$ is an infinitary operator.

Operators $\delta$ and $\bigoplus_{i\in I}[p_i]\textvisiblespace$ are
used to construct discrete probability functions of countable support:
$\delta(t)$ is interpreted as a distribution that assigns probability
1 to the state term $t$ and probability 0 to any other term $t'$ (syntactically)
different from $t$, and
$\bigoplus_{i\in I}[p_i]\theta_i$ represents a distribution that weights
with $p_i$ the distribution represented by the term $\theta_i$.
Moreover, a probabilistically lifted operator $\dist{f}$ is
interpreted by properly lifting the probabilities of the operands to
terms composed with the operator $f$.

Formally, the algebra associated with a probabilistically lifted
signature $\Sigma = (\Ssignature,\Dsignature)$ is defined as follows.
For sort $\TS$, it is the freely generated algebraic structure
$\closedSTerms$.  
For sort $\DTS$, it is defined by the carrier
$\Delta(\closedSTerms)$ and the following interpretation:
  $\sem{\delta(t)}{} = \delta_t$ for all $t\in \closedSTerms$,
  $\sem{\bigoplus_{i\in I}[p_i]\theta_i}{} =
     \sum_{i\in I}p_i \sem{\theta_i}{}$
    for $\{\theta_i\mid i\in I\} \subseteq \closedDTerms$,
  $\sem{\dist{f}(\theta_1,\ldots,\theta_{\rank(f)})}{}(f(\xi_1,\ldots,\xi_{\rank(f)})) = \prod_{\sigma_i=s} \sem{\theta_i}{}(\xi_i)$ if for all $j$ s.t.\ $\sigma_j = d$, $\theta_j = \xi_j$, and
  $\sem{\dist{f}(\theta_1,\ldots,\theta_{\rank(f)})}{}(f(\xi_1,\ldots,\xi_{\rank(f)})) = 0$ otherwise.
Here it is assumed that $\prod\emptyset = 1$.
Notice that in the semantics of a lifted function $\dist{f}$, the big
product only considers the distributions related to the \TS-sorted
positions in $f$, while the distribution terms corresponding to the
\DTS-sorted positions in $f$ should match exactly to the
parameters of $\dist{f}$.

A \emph{substitution} $\rho$ is a map
$\TVar \cup \DVar \to \openTerms$ such that 
$\rho(x) \in \openSTerms$, for all $x\in\TVar$, and
$\rho(\mu) \in \openDTerms$, for all $\mu\in\DVar$.
A substitution is closed if it maps each
variable to a closed term. A substitution extends to a mapping from
terms to terms as usual.

Finally, we remark a general property of distribution terms: let
$f\in\Ssignature$ with $\arity(f) = \sigma_1 \ldots \sigma_n \to \TS$,
and let $\sigma_j = \TS$; then $\dist{f} \in \Dsignature$ is distributive
w.r.t.\ $\oplus$ in the position $j$, i.e.\
 $\sem{\rho(\dist{f}(\ldots,\xi_{j-1},\bigoplus_{i\in I}[p_i]\theta_i,\xi_{j+1},\ldots))}{} =
\sem{\rho(\bigoplus_{i\in I}[p_i]\dist{f}(\ldots,\xi_{j-1},\theta_i,\xi_{j+1},\ldots))}{}$
for any closed substitution $\rho$.  The proof follows
from the definition of $\sem{\textvisiblespace}{}$.
However, notice that $\dist{f}$ \emph{does not} distribute
w.r.t.\ $\oplus$ in a position $k$ such that $\sigma_k=\DTS$.


\section{Probabilistic Transition System Specifications}\label{sec:ptss}%

A (probabilistic) transition relation prescribes which possible
activity can be performed by a term in a signature.  Such activity is
described by the label of the action and a probability distribution on
terms that indicates the probability to reach a particular new term.
We will follow the probabilistic automata style of probabilistic
transitions~\cite{Seg95a} which are a generalization of the so-called reactive model~\cite{LS91}.

\begin{definition}[PTS]%
  A \emph{probabilistic labeled transition system} (PTS) is a triple
  $(\closedSTerms,\Act,{\trans})$, where $\Sigma =
  (\Ssignature,\Dsignature)$ is a probabilistically lifted signature,
  $\Act$ is a countable set of actions,
  and ${\trans} \subseteq \closedSTerms\times\Act\times\Delta(\closedSTerms)$,
  is a transition relation.
We write $t\trans[a]\pi$ for $(t,a,\pi)\in{\trans}$.
\end{definition}

Transition relations are usually defined by means of structured
operational semantics in Plotkin's style~\cite{Plo04}.
For PTS, algebraic characterizations of this style were provided in~\cite{DL12, LGD12,DGL15a} 
where the term \emph{probabilistic transition system specification} was used and which we adopt in
our paper.

\begin{definition}[PTSS]\label{def:ptss}%
  A \emph{probabilistic transition system specification} (PTSS) is a
  triple $P = (\Sigma, \Act, R)$ where
  $\Sigma$ is a probabilistically lifted signature,
  $\Act$ is a set of labels, and $R$ is a set of rules of the form:
  \[
  \ddedrule{ \{t_k \trans[a_k] \theta_k \mid k\in K \}\cup
              \{t_l \ntrans[b_l] \mid l \in L\} \cup
              \{\theta_j (T_j) \bowtie_j q_j \mid j \in J\} } %
            { t \trans[a] \theta}
  \]
  where
  $K, L, J$ are index sets, 
  $t, t_k, t_l \in \openSTerms$, $a, a_k, b_l \in \Act$, 
  $T_j\subseteq \closedSTerms$,
  ${\bowtie_j} \in \{{>},{\geq}, <, \leq \}$, $q_j\in[0,1]$
  and $\theta_j, \theta_k, \theta \in \openDTerms$.
\end{definition}

Expressions of the form $t \trans[a] \theta$, $t \ntrans[a]$, and
$\theta (T) \bowtie p$ are called \emph{positive literal},
\emph{negative literal}, and \emph{quantitative literal},
respectively.
For any rule $r \in R$, literals above the line are called
\emph{premises}, notation $\prem{r}$; the literal below the line is
called \emph{conclusion}, notation $\conc{r}$.
We denote with $\pprem{r}$, $\nprem{r}$, and $\qprem{r}$ the sets of
positive, negative, and quantitative premises of the rule $r$,
respectively.
In general, we allow the sets of positive, negative, and quantitative premises to be infinite.

Substitutions provide instances to the rules of a PTSS that, together
with some appropriate machinery, allow us to define probabilistic
transition relations.  Given a substitution $\rho$, 
it extends to literals as follows:
$\rho(t \ntrans[a]) = \rho(t) \ntrans[a]$,
$\rho(\theta(T) \ {\bowtie} \ p) = \rho(\theta)(\rho(T)) \ {\bowtie} \ p$
(where $\rho(T)= \{\rho(t) \mid t \in T\}$),
and
$\rho(t \trans[a] \theta) = \rho(t) \trans[a] \rho(\theta)$.

We say that $r'$ is a (closed) instance of a rule $r$ if there is a
(closed) substitution $\rho$ so that $r'=\rho(r)$.
We say that $\rho$ is a \emph{proper substitution of $r$} if for all
quantitative premises $\theta(T) \bowtie p$ of $r$ and all $t\in T$,
$\sem{\rho(\theta)}{}(\rho(t)) > 0$ holds.  
We use only this kind of substitution in the paper. 

In the rest of the paper, we will deal with models as \emph{symbolic}
transition relations in the set
$\closedSTerms\times\Act\times\closedDTerms$ rather than the
\emph{concrete} transition relations in
$\closedSTerms\times\Act\times\Delta(\closedSTerms)$ required by a
PTS.
Hence we will mostly refer with the term ``transition relation'' to
the symbolic transition relation.
In any case, a symbolic transition relation induces always a unique
concrete transition relation by interpreting every target distribution
term as the distribution it defines; that is, the symbolic transition
$t\trans[a]\theta$ is interpreted as the concrete transition
$t\trans[a]\sem{\theta}{}$.  If the symbolic transition relation turns out
to be a model of a PTSS $P$, we say that the induced concrete
transition relation defines a PTS associated to $P$.

To define an appropriate notion of model we consider \emph{3-valued models}.  
A 3-valued model partitions the set $\closedSTerms\times\Act\times\closedDTerms$
in three sets containing, respectively, the transition that are known
to hold, that are known not to hold, and those whose validity is
unknown.
Thus, a 3-valued model can be presented as a pair
$\tuple{\CerTr,\PosTr}$ of transition relations
$\CerTr,\PosTr\subseteq{\closedSTerms\times\Act\times\closedDTerms}$,
with $\CerTr\subseteq\PosTr$, where $\CerTr$ is the set of transitions
that \emph{certainly} hold and $\PosTr$ is the set of transitions that
\emph{possibly} hold.  So, transitions in $\PosTr\setminus\CerTr$ are
those whose validity is unknown and transitions in
$(\closedSTerms\times\Act\times\closedDTerms)\setminus\PosTr$ are
those that certainly do not hold.
A 3-valued model $\tuple{\CerTr,\PosTr}$ that is justifiably compatible
with the proof system defined by a PTSS $P$ is said to be \emph{stable}
for $P$.
(See Def.~\ref{def:3-valued-stable-model}.)

Before formally defining the notions of proof and 3-valued stable
model we introduce some notation.
Given a transition relation
$\Tr \subseteq \closedSTerms\times\Act\times\closedDTerms$,
$t\trans[a]\theta$ \emph{holds in} $\Tr$, notation $\Tr\models
{t\trans[a]\theta}$, if ${t\trans[a]\theta} \in \Tr$;
$t\ntrans[a]$ \emph{holds in} $\Tr$, notation $\Tr\models
{t\ntrans[a]}$, if for all $\theta\in\closedDTerms$,
${t\trans[a]\theta} \notin \Tr$.
A closed quantitative constraint $\theta(T)\bowtie p$ \emph{holds in}
$\Tr$, notation $\Tr\models \theta(T)\bowtie p$, if
$\sem{\theta}{}(T)\bowtie p$.  Notice that the satisfaction of a
quantitative constraint does not depend on the transition relation. We
nonetheless use this last notation as it turns out to be convenient.
Given a set of literals $H$, we write $\Tr \models H$
if for all $\phi \in H$, $\Tr \models \phi$.

\begin{definition}[Proof]\label{def:proof}
  Let $P = (\Sigma, \Act, R)$ be a PTSS.
  Let $\psi$ be a positive literal and let $H$ be a set of literals. 
  A \emph{proof} of a transition rule $\frac{H}{\psi}$ from $P$ is a
  well-founded, upwardly branching tree where each node is a literal
  such that:
  \begin{inparaenum}[(i)]
  \item%
    the root is $\psi$; and 
  \item%
    if $\chi$ is a node and $K$ is the set of nodes directly above
    $\chi$, then one of the following conditions holds:
    \begin{inparaenum}
    \item%
      $K = \emptyset$ and $\chi \in H$, or
    \item%
      $\chi = (\theta(T)\bowtie p)$ is a closed quantitative literal
      such that $\sem{\theta}{}(T)\bowtie p$ holds, or
    \item%
      $\frac{K}{\chi}$ is a valid substitution instance of a rule from
      $R$.
    \end{inparaenum}
  \end{inparaenum}

  $\dedrule{H}{\psi}$ \emph{is provable from} $P$, notation $P \proves
  \dedrule{H}{\psi}$, if there exists a proof of $\dedrule{H}{\psi}$
  from $P$.
\end{definition}

Before, we said that a 3-valued stable model $\tuple{\CerTr,\PosTr}$ for
a PTSS $P$ has to be \emph{justifiably compatible} with the proof
system defined by $P$.  By ``compatible'' we mean that
$\tuple{\CerTr,\PosTr}$ has to be consistent with every provable rule.
With ``justifiable'' we require that for each transition in $\CerTr$
and $\PosTr$ there is actually a proof that justifies it.
More precisely, we require that
\begin{inparaenum}[(a)]
\item%
  for every certain transition in $\CerTr$ there is a proof in $P$
  such that all negative hypotheses of the proof are known to hold
  (i.e. there is no possible transition in $\PosTr$ denying a negative
  hypothesis), and
\item%
  for every possible transition in $\PosTr$ there is a proof in $P$
  such that all negative hypotheses possibly hold (i.e. there is no
  certain transition in $\CerTr$ denying a negative hypothesis).
\end{inparaenum}
This is formally stated in the next definition.

\begin{definition}[3-valued stable model]\label{def:3-valued-stable-model}%
  Let $P = (\Sigma, \Act, R)$ be a PTSS.
  A tuple $\tuple{\CerTr,\PosTr}$ with
  $\CerTr\subseteq\PosTr\subseteq{\closedSTerms\times\Act\times\closedDTerms}$
  is a \emph{3-valued stable model} for $P$ if for every closed
  positive literal $\psi$,
  \begin{compactenum}[(a)]
  \item%
    $\psi\in\CerTr$ iff there is a set $N$ of closed negative literals
    such that $P \proves \dedrule{N}{\psi}$ and $\PosTr\models N$
  \item%
    $\psi\in\PosTr$ iff there is a set $N$ of closed negative literals
    such that $P \proves \dedrule{N}{\psi}$ and $\CerTr\models N$.
  \end{compactenum}
\end{definition}

The least 3-valued stable model of a PTSS can be
constructed using induction \cite{DGL15a, Fok00,FV98}.

\begin{lemma}\label{lem:inductive-construction-of-3v-least-model}
  Let $P$ be a PTSS.
  For each ordinal $\alpha$, define the pair
  $\tuple{\CerTr_\alpha,\PosTr_\alpha}$ as follows:
  \begin{itemize}
  \item%
    $\CerTr_0=\emptyset$ and
    $\PosTr_0=\closedSTerms\times\Act\times\closedDTerms$.
  \item%
    For every non-limit ordinal $\alpha>0$, define:
    \begin{align*}
      \CerTr_\alpha & \textstyle
      = \Big\{t\trans[a]\theta \mid
               \text{for some set } N \text{ of negative literals, }
               P\proves\dedrule{N}{t\strans[a]\theta}
               \text{ and }
               \PosTr_{\alpha-1}\models N \Big\}
      \\
      \PosTr_\alpha & \textstyle
      = \Big\{t\trans[a]\theta \mid
               \text{for some set } N \text{ of negative literals, }
               P\proves\dedrule{N}{t\strans[a]\theta}
               \text{ and }
               \CerTr_{\alpha-1}\models N \Big\}
    \end{align*}
  \item%
    For every limit ordinal $\alpha$, define
    $
      \CerTr_\alpha = \bigcup_{\beta<\alpha}\CerTr_\beta
    $ and
    $
      \PosTr_\alpha = \bigcap_{\beta<\alpha}\PosTr_\beta
    $.
  \end{itemize}
  Then:
\begin{inparaenum}[1.]
  \item\label{lem:inductive-construction-of-3v-least-model:inclusion}%
    if $\beta\leq\alpha$, $\CerTr_\beta\subseteq\CerTr_\alpha$ and
    $\PosTr_\beta\supseteq\PosTr_\alpha$, 
    and
  \item\label{lem:inductive-construction-of-3v-least-model:least}%
    there is an ordinal $\lambda$ such that
    $\CerTr_\lambda=\CerTr_{\lambda+1}$ and
    $\PosTr_\lambda=\PosTr_{\lambda+1}$.
    Moreover, $\tuple{\CerTr_\lambda,\PosTr_\lambda}$ is the least
    3-valued stable model for $P$.
\end{inparaenum}
\end{lemma}

PTSSs with least 3-valued stable model that are also a 2-valued model
are particularly interesting, since this model is actually the only
3-valued stable model~\cite{BG96,vG04}.
  A PTSS $P$ is said to be \emph{complete} if its least 3-valued
  stable model $\tuple{\CerTr,\PosTr}$ satisfies that $\CerTr=\PosTr$
  (i.e., the model is also 2-valued).
We associate a probabilistic transition system to each
complete PTSS.
\begin{definition}\label{def:associated-model}%
  Let $P$ be a complete PTSS and let $\tuple{\Tr,\Tr}$ be its unique
  3-valued stable model.
  We say that $\Tr$ is the \emph{transition relation associated to}
  $P$.  We also define the \emph{PTS associated to} $P$ as the unique
  PTS $(\closedSTerms,\Act,{\trans})$ such that
  $t\trans[a]\pi$ if and only if
  $t\trans[a]\theta\in\Tr$ and $\sem{\theta}{}=\pi$ for some
  $\theta\in\closedDTerms$.
\end{definition}

The different examples that we give in the rest of the papers are in
terms of a basic probabilistic process algebra.  We introduce it here,
but address the reader to~\cite{DGL15a} for an example of a PTSS with
richer operators.
Signature $\Ssignature$ contains the constant $\nullproc$,
representing the stop process, for each action $a\in\Act$, a unary
probabilistic prefix operators $a.\textvisiblespace$ with arity
$\arity(a) = \DTS \to \TS$, and a binary operator $+$, the alternative
composition or sum, with arity $\arity({+}) = {\TS}{\TS} \to \TS$,
while $\Dsignature$ contains the respective lifted signature,
$\delta$, and all binary operators
$\bigoplus_{i\in\{1,2\}}[p_i]\theta_i$ which we denote by $\oplus_p$.
The semantics is defined with the usual rules:
{\small%
\[\dedrule{}{a.\mu \trans[a] \mu}
  \qquad\qquad
  \dedrule{x\trans[a]\mu}{{x+y}\trans[a]\mu}
  \qquad\qquad
  \dedrule{y\trans[a]\mu}{{x+y}\trans[a]\mu}
\]%
}%

\section{Bisimulation relations}\label{sec:semantic_relations}

This work revolves around four different types of bisimulation relations:
\begin{inparaenum}[(i)]
\item%
  the usual \emph{(strong) bisimulation}~\cite{LS91} relation on
  probabilistic system, in which each probabilistic transition should
  be matched with a single probabilistic transition so that the
  distributions of both transitions agree on the probabilities of
  jumping into equivalent states;
\item%
  the \emph{convex bisimulation}~\cite{Seg95a} relation, in which the
  matching is performed instead with a convex combination of
  transition relations;
\item%
  the \emph{probability abstracted bisimulation}, in which the
  matching is performed by a single transition so that the
  distributions of both transitions agree on jumping to the same
  equivalent classes of states but not necessarily with the same
  probability value; and
\item%
  the \emph{probabilistic obliterated bisimulation}, which represents
  the usual bisimulation~\cite{Mil89} once the probabilistic
  transition system is abstracted into a traditional labeled transition
  system in the usual way.
\end{inparaenum}

\remarkMDL{This paragraph is related with the full abstraction section
  that is not present anymore in the paper.}
\remarkPRD{But it was related to the logic characterization part that it is now on this section}
To our knowledge, the probability abstracted bisimulation originates
here.  Its intention is to strictly preserve the probabilistic
structure of a system without caring about the probability values.  Thus,
probability abstracted bisimulation is consistent with any
bisimulation preserving quantitative properties that only tests for
positive quantifications, rather than a particular value.
Instead, this kind of properties are not preserved by the
probabilistic obliterated bisimulation as it is shown below in
this section.

In the following we introduce all these relations and discuss their
relationship as well as alternative definitions.  For the rest of the
section we assume given a PTS $P= (\closedSTerms,\Act,{\trans})$.

Given a relation ${\relR}\subseteq{\closedSTerms\times\closedSTerms}$,
a set $Q \subseteq \closedSTerms$ is \emph{$\closed{\relR}$} if for
all $t \in Q$ and $t' \in \closedSTerms$, $t \relR t'$ implies $t' \in
Q$ (i.e. ${\relR}(Q) \subseteq Q$).
It is easy to verify that if two relations
${\relR}, {\relR'} \subseteq \closedSTerms \times \closedSTerms$
are such that ${\relR'}\subseteq{\relR}$, then
if $Q \subseteq \closedSTerms$ is $\closed{\relR}$, it is also $\closed{\relR'}$.

\begin{definition}\label{def:bisimulation}
  A relation ${\relR} \subseteq {\closedSTerms \times \closedSTerms}$
  is a \emph{bisimulation} if it is symmetric and for all $t, t' \in
  \closedSTerms$, $a\in \Act$, and $\pi \in \Delta(\closedSTerms)$,
    $t \relR t'$ and $t \trans[a] \pi$ imply that there exists $\pi'
    \in \Delta(\closedSTerms)$ s.t.\ $t' \trans[a] \pi'$ and $\pi
    \relR \pi'$,
  where $\pi \relR \pi'$ if and only if for all $\closed{\relR}\ Q
  \subseteq \closedSTerms$, $\pi(Q) = \pi'(Q)$.
  The relation $\bisim$, called \emph{bisimilarity} or
  \emph{bisimulation equivalence}, is defined as the smallest relation
  that includes all bisimulations.
\end{definition}

A \emph{combined transition} $t \trans[a]_c \pi$ is defined whenever
there is a family $\{\pi_i\}_{i\in I} \subseteq \Delta(\closedSTerms)$
and a family $\{p_i\}_{i\in I} \subseteq [0,1]$ such that $t \trans[a]
\pi_i$ for all $i\in I$, $\sum_{i\in I}p_i = 1$ and $\pi=\sum_{i\in I}p_i\pi_i$.

\begin{definition}
  \label{def:convex_bisimulation}%
  A relation ${\relR} \subseteq {\closedSTerms \times \closedSTerms}$
  is a \emph{convex bisimulation} if it is symmetric and for all $t,
  t' \in \closedSTerms$, $a\in \Act$, and $\pi \in
  \Delta(\closedSTerms)$,
    $t \relR t'$ and $t \trans[a] \pi$ imply that there exists $\pi'
    \in \Delta(\closedSTerms)$ s.t.\ $t' \trans[a]_c \pi'$ and $\pi
    \relR \pi'$.
  The relation $\bisim_c$, called \emph{convex bisimilarity} or
  \emph{convex bisimulation equivalence}, is defined as the smallest
  relation that includes all convex bisimulations.
\end{definition}

\begin{definition}
  \label{def:p_abstracted_bisimulation}%
  A relation ${\relR} \subseteq {\closedSTerms \times \closedSTerms}$
  is a \emph{probability abstracted bisimulation} if it is symmetric
  and for all $t, t' \in \closedSTerms$, $a\in \Act$, and $\pi \in
  \Delta(\closedSTerms)$,
    $t \relR t'$ and $t \trans[a] \pi$ imply that there exists $\pi'
    \in \Delta(\closedSTerms)$ s.t.\ $t' \trans[a] \pi'$ and for all
    $\closed{\relR}$ $Q \subseteq \closedSTerms$, $\pi(Q) > 0$ iff
    $\pi'(Q) > 0$.
  The relation $\bisim_a$, called \emph{probability abstracted
    bisimilarity} or \emph{probability abstracted bisimulation
    equivalence}, is defined as the smallest relation that includes
  all probability abstracted bisimulations.
\end{definition}

Notice that the transfer property in this last case follows the same
structure as the bisimulation, only that it only requires that
$\pi(Q)>0$ iff $\pi'(Q)>0$ for all $\closed{R}$, instead of
$\pi(Q)=\pi'(Q)$.

\begin{definition}
  \label{def:p_obliterated_bisimulation}%
  A relation ${\relR} \subseteq {\closedSTerms \times \closedSTerms}$
  is a \emph{probability obliterated bisimulation} if it is symmetric
  and for all $t, t' \in \closedSTerms$, $a\in \Act$, and $\pi \in
  \Delta(\closedSTerms)$,
    $t \relR t'$ and $t \trans[a] \pi$ imply that for all
    $\closed{\relR}$ $Q \subseteq \closedSTerms$ with $\pi(Q) > 0$,
    there exists $\pi' \in \Delta(\closedSTerms)$ s.t.\ $t' \trans[a]
    \pi'$ and $\pi'(Q) > 0$.
  The relation $\bisim_o$, called \emph{probability obliterated
    bisimilarity} or \emph{probability obliterated bisimulation
    equivalence}, is defined as the smallest relation that includes
  all probability obliterated bisimulations.
\end{definition}

Compare this last definition with
Def.~\ref{def:p_abstracted_bisimulation}.  While for probability
abstracted bisimulation we require that there is a single matching
transition $t' \trans[a] \pi'$ so that $\pi'$ gives some positive
probability to all $\closed{R}$ sets exactly whenever $\pi$ does, the
definition of probability obliterated bisimulation permits to choose
different matching transitions for each $\closed{R}$ set that measures
positively on $\pi$.

It is well known that $\bisim$ and $\bisim_c$ are equivalences
relations and that they also are, respectively, a bisimulation
relation and a convex bisimulation relation.  The fact that $\bisim_o$
is also an equivalence relation and itself a probability obliterated
bisimulation follows from Lemma~\ref{lemma:pob_using_otrans} which
state that it agrees with Park \& Milner's bisimulation.
The same properties can be proven for probability abstracted
bisimulation:

\begin{lemma}\label{lemma:p_abst_bisim_equiv}
  $\bisim_a$ is an equivalence relation and is itself a probability
  abstracted bisimulation.
\end{lemma}

Similarly to the bisimulation~\cite[Prop 3.4.4]{Bai98}, the
probability abstracted bisimulation has a characterization in terms of
an \emph{abstract} weight function.  This alternative characterization
is the one used in the proof of
Theorem~\ref{th:prob_abstracted_format} and that is why we present it
in this paper.

Given a relation ${\relR} \subseteq {\closedSTerms \times
  \closedSTerms}$, we define $\equiv^w_{\relR} \in
{\Delta(\closedSTerms) \times \Delta(\closedSTerms)}$ as follows. For
all $\pi,\pi'\in\Delta(\closedSTerms)$, $\pi \equiv^w_{\relR} \pi'$ if
there is an \emph{abstract weight function}
$w:(\closedSTerms\times\closedSTerms)\to[0,1]$ s.t.\ for all
$t,t'\in\closedSTerms$,
\begin{inparaenum}[(i)]
\item\label{eq:weigh-left}%
  $w(t,\closedSTerms) > 0$ iff $\pi(t) > 0$,
\item\label{eq:weigh-right}%
  $w(\closedSTerms,t') > 0$ iff $\pi'(t') > 0$, and
\item\label{eq:weigh-rel}
  $w(t,t') > 0$ implies $t \relR t'$.
\end{inparaenum}

\begin{lemma}\label{lemma:altdef_p_abst_bisim}
  For all $t, t' \in \closedSTerms$, $t\bisim_a t'$ if and only if
  there is a symmetric relation ${\relR} \subseteq
  {\closedSTerms\times\closedSTerms}$ with $t \relR t'$ such that for
  all $t_1, t_2 \in \closedSTerms$, $a\in \Act$, and $\pi_1 \in
  \Delta(\closedSTerms)$,
    $t_1 \relR t_2$ and $t_1 \trans[a] \pi_1$ imply that there exists
    $\pi_2 \in \Delta(\closedSTerms)$ s.t.\ $t_2 \trans[a] \pi_2$ and
    $\pi_1 \equiv^w_{\relR} \pi_2$.
\end{lemma}

The next lemma shows that the probability obliterated bisimulation
agrees with Park \& Milner's bisimulation.
Denote $t\otrans[a]t'$ iff there is $\pi$ such that $t\trans[a]\pi$
and $\pi(t')>0$.  Notice that this notation precisely defines the
usual abstraction of probabilistic transition systems into labeled
transition systems in which all information regarding the probability
distribution is lost except from the fact that one state can reach
another state with positive probability after a transition.

\remarkPRD{Since I do not know how will finally the proof of the
  format for probability obliterated bisimulation will end up, I
  decided to keep this only alternative characterization which is
  used in the proof of the logic characterization. other alternative
  characterizations are commented out and we may later include them}

\begin{lemma}\label{lemma:pob_using_otrans}
  For all $t, t' \in \closedSTerms$, $t\bisim_o t'$ iff
  there is a symmetric relation ${\relR} \subseteq
  {\closedSTerms\times\closedSTerms}$ with $t \relR t'$ s.t.\ for
  all $t_1,t_2,t'_1 \in \closedSTerms$ and $a\in \Act$,
    $t_1 \relR t_2$ and $t_1 \otrans[a] t'_1$ imply that there exists
    $t'_2 \in \closedSTerms$ s.t.\ $t_2 \otrans[a] t'_2$ and $t'_1
    \relR t'_2$.
\end{lemma}

Finally we state the relation among the different bisimulations

\begin{lemma}\label{lemma:relations_among_sem}
  The following inclusions hold and are proper:
  ${\bisim} \subsetneq {\bisim_c} \subsetneq {\bisim_o}$ and
  ${\bisim} \subsetneq {\bisim_a} \subsetneq {\bisim_o}$.
  Besides ${\bisim_c}$ and ${\bisim_a}$ are incomparable.
\end{lemma}

In fact the results can be proved stronger as we explain in the following.
Any bisimulation relation is also a convex bisimulation, which follow
from the fact that $t\trans[a]\pi$ implies $t\trans[a]_c\pi$.
Any convex bisimulation is also a probability obliterated bisimulation
since $t\trans[a]_c\pi$ with $\pi(Q)>0$ implies that there is a $\pi'$
such that $t\trans[a]\pi'$ and $\pi'(Q)>0$.
Any bisimulation is also a probability abstracted bisimulation since
$\pi \relR \pi'$ implies $\pi(Q) > 0$ iff $\pi'(Q) > 0$ for all
$\closed{\relR}$ $Q$.
Finally, any probability abstracted bisimulation is also a probability
obliterated bisimulation since, for a given $\pi$ and $\relR$, the
existence of a $\pi'$ s.t.\ $t'\trans[a]\pi'$ and $\pi(Q)>0$ iff
$\pi'(Q)>0$ for all $\closed{\relR}$ $Q$, guarantees that, for all
$\closed{\relR}$ $Q$ with $\pi(Q)>0$ there is a $\pi'$
s.t.\ $t'\trans[a]\pi'$ and $\pi'(Q)>0$.

Notice that $\te_1=a.(\dist{b.\nullproc})+a.(\dist{c.\nullproc})$ and
$\te_2 = \te_1+a.(\dist{b.\nullproc}\oplus_{0.5}\dist{c.\nullproc})$
are convex bisimilar but not probability abstracted bisimilar.
Besides, notice that
$\te_3=a.(\dist{b.\nullproc}\oplus_{0.5}\dist{c.\nullproc})$ and
$\te_4=a.(\dist{b.\nullproc}\oplus_{0.1}\dist{c.\nullproc})$ are probability
abstracted bisimilar but not convex bisimilar.  These examples not
only show that ${\bisim_c}$ and ${\bisim_a}$ are incomparable, but
also that all stated inclusions are proper.

In the rest of the section we present logical characterizations for
the different bisimulation equivalences.  This work has already been
done for bisimulation~\cite{HPS11,DTW12} and convex
bisimulation~\cite{HPS11}.  We adopt here the two-level logic style
of~\cite{DTW12}.
  
We define the logic $\blogic$ as the set of all formulas with the
following syntax:
\[\textstyle
  \phi := \ \top \ \mid \ \dia{a}{\psi} \ \mid \ \diac{a}{\psi} \ \mid \ \bigwedge_{i \in I} \phi_i \ \mid \ \neg \phi
  \qquad\qquad
  \psi :=  \ \textstyle [\phi]_{p} \ \mid \ \pand_{i \in I} \psi_i
\]
where $a \in A$, $p \in [0,1]\cap\Q$, and $I$ is any index set.
The logic $\clogic$ contains all formulas of $\blogic$ without the
modality $\dia{a}{\_}$ .
The logic $\alogic$ contains all formulas of $\blogic$ without
modalities $\diac{a}{\_}$ and $[\_]_{p}$ for all $p>0$ (i.e.\ it only
accepts $[\_]_0$ among this type of modalities.)
Finally, the logic $\ologic$ contains all formulas of $\alogic$ without
$\pand_{i \in I}\_$ .

The semantics of $\blogic$ is defined with the satisfaction relation
$\models$ on a PTS $P = (\closedSTerms, A, \trans)$ as follows.
\begin{align*}
  \ctop && t & \models \top &&\text{ for all $t \in \closedSTerms$ } &
  \cneg && t & \models \neg \phi &&\text{ if $t \not \models \phi$}\\
  \ca && t & \models \dia{a}{\psi} &&\text{ if there is $t \trans[a] \pi$ s.t.\ $\pi \models \psi$} &
  \cmeasure && \pi & \models [\phi]_{p} &&\text{ if $\pi(\{t \in \closedSTerms \mid t \models \phi\}) > p$}\\
  \cac && t & \models \diac{a}{\psi} &&\text{ if there is $t \trans[a]_c \pi$ s.t.\ $\pi \models \psi$} & \qquad
  \textstyle\cpand && \pi & \models \textstyle \pand_{i \in I} \psi_i &&\text{ if $\pi \models \psi_i$ for all $i \in I$}\\
  \textstyle\cland && t & \models \textstyle \bigwedge_{i \in I} \phi_i &&\text{ if $t \models \phi_i$ for all $i \in I$ }
\end{align*}
The semantics of the other logics is defined in the same way but
restricted to the respective operators.

For $\chi\in\{b,c,a,o\}$, let
$\xlogic(t) = \{{\phi\in\xlogic} \mid t \models \phi\}$, for all
$t\in\closedSTerms$, and
$\xlogic(\pi) = \{{\psi\in\xlogic} \mid \pi \models \psi\}$, for all
$\pi\in\Delta(\closedSTerms)$.
We write $t_1 \bisim_{\xlogic} t_2$ iff $\xlogic(t_1) = \xlogic(t_2)$
and $\pi_1 \bisim_{\xlogic} \pi_2$ iff $\xlogic(\pi_1) = \xlogic(\pi_2)$.
Then, we have the following characterization theorem.

\begin{theorem}\label{th:logic_characterization}
  For all $\chi\in\{b,c,a,o\}$ and for all $t_1, t_2 \in
  \closedSTerms$, $t_1 \bisim_\chi t_2$ iff $t_1 \bisim_{\xlogic} t_2$
  (where ${\bbisim}={\bisim}$).
\end{theorem}

Let $\te_1$, $\te_2$, $\te_3$, and $\te_4$ be as before.  Recall
$\te_1\cbisim\te_2$ and $\te_3\abisim\te_4$.
Notice that
$\dia{a}{([\dia{b}{\top}]_{0.5}\sqcap[\dia{c}{\top}]_{0.5})}$
distinguish $\te_1$ from $\te_2$, while
$\diac{a}{([\dia{b}{\top}]_{0.5}\sqcap[\dia{c}{\top}]_{0.5})}$ is
satisfied by both $\te_1$ and $\te_2$.  That is why $\dia{a}{\_}$ is not
an operator of $\clogic$.
Notice $[\dia{b}{\top}]_{0.5}$ distinguishes the distribution
$\sem{b.\dist{\nullproc}\oplus_{0.5}c.\dist{\nullproc}}{}$ from
$\sem{b.\dist{\nullproc}\oplus_{0.1}c.\dist{\nullproc}}{}$, while
$[\dia{b}{\top}]_0$ does not
\remarkMDL{I do not understand the sentence: \emph{but it does distinguish them from}}
\remarkPRD{Check the whole sentencenow thatI changed it}
(but it does distinguish them from e.g.\ $\sem{c.\dist{\nullproc}}{}$).
Thus $\dia{a}{[\dia{b}{\top}]_{0.5}}$ distinguishes $\te_3$ from
$\te_4$.  That is why $[\_]_p$ is not an operator of $\alogic$ if
$p>0$.
Finally, notice that
$\dia{a}{([\dia{b}{\top}]_0\sqcap[\dia{c}{\top}]_0)}$ distinguishes
$\te_5 = a.(\dist{b.\nullproc} \oplus_{0.5} \dist{c.\nullproc})$ from
$\te_6 = a.\dist{b.\nullproc} + a.\dist{c.\nullproc}$, and observe
that $\te_5 \bisim_o \te_6$.
However, neither $\dia{a}{[\dia{b}{\top}]_0}$ nor
$\dia{a}{[\dia{c}{\top}]_0}$ can distinguish them.  That is why
$\pand_{i\in I}\_$ is not an operator of $\ologic$.

\section{Formats}\label{sec:formats}

In this section we introduce rule and specification formats that
guarantee that each bisimulation equivalences discussed in the
previous section is a congruence for every operator whose semantics is
defined within the respective rule of the specification format.
In particular, the format \ntmufxt, which ensures that bisimulation
equivalence is a congruence for all operator in such format, has been
already introduced in~\cite{DL12} and finally revised
in~\cite{DGL15a}.  We present here its more general form.

The following definition is important to ensure a symmetric treatment
of variables and terms within the format.
\label{page:diag}
Let $\{Y_l\}_{l\in L}$ be a family of sets of state term variables
with the same cardinality.  The $l$-th element of a tuple $\vec{y}$ is
denoted by $\vec{y}(l)$.  For a set of tuples $T=\{\vec{y_i} \mid i\in
I\}$ we denote the $l$-th projection by $\proj_l(T)=\{\vec{y_i}(l)
\mid i\in I\}$.
Fix a set $\diag{Y_l}_{l\in L} \subseteq \prod_{l\in L} Y_l $ such that:
\begin{inparaenum}[(i)]
\item%
  for all $l \in L$, $\proj_l(\diag{Y_l}_{l\in L}) = Y_l$; and
\item%
  for all $\vec{y}, \vec{y'} \in \diag{Y_l}_{l\in L}$, $(\exists l \in L :
  \vec{y}(l) = \vec{y'}(l)) \limp \vec{y} = \vec{y'}$.
\end{inparaenum}
Property (ii) ensures that different tuples
$\vec{y},\vec{y'}\in\diag{Y_l}_{l\in L}$ differ in all positions, and
by property (i) every variable of every $Y_l$ is used in (exactly) one
$\vec{y} \in \diag{Y_l}_{l\in L}$.
$\Diag$ stands for ``diagonal'', following the intuition that each
$\vec{y}$ represents a coordinate in the space $\prod_{l\in L} Y_l$,
so that $\diag{Y_l}_{l\in L}$ can be seen as the line that traverses
the main diagonal of the space.
Therefore, notice that, for $Y_l=\{y_l^0,y_l^1,y_l^2,\ldots\}$, a
possible definition for the set $\diag{Y_l}_{l\in L}$ is
$  \{(y_1^0,y_2^0,\ldots,y_L^0),(y_1^1,y_2^1,\ldots,y_L^1),(y_1^2,y_2^2,\ldots,y_L^2),\ldots\}
$.
In addition, we use the following notation:
$t(\zeta_1,\ldots,\zeta_n)$ denotes a term that only has variables in
the set $\{\zeta_1,\ldots,\zeta_n\}$, that is
$\Var(t(\zeta_1,\ldots,\zeta_n))\subseteq\{\zeta_1,\ldots,\zeta_n\}$,
and moreover, $t(\zeta'_1,\ldots,\zeta'_n)$ denotes the same term as
$t(\zeta_1,\ldots,\zeta_n)$ in which each variable $\zeta_i$ has been
replaced by $\zeta'_i$.

\begin{definition}\label{def:ntmufxtheta}
  Let $P = (\Sigma, \Act, R)$ be a PTSS.
  A rule $r\in R$ is in \emph{\ntmuft\ format} 
  if it has the following form %
  \[
  \ddedrule%
      {
	{\textstyle
         \bigcup_{m\in M}
           \{ t_m(\vec{z})\trans[a_m]\mu_m^{\vec{z}} \mid \vec{z}\in \mathcal{Z} \} \ \cup \
         \bigcup_{n\in N}
           \{ t_n(\vec{z})\ntrans[b_n] \mid \vec{z}\in \mathcal{Z} \} \  \cup \
	  }
         {\textstyle
		\{ \theta_l(Y_l)\gtgeq_{l,k} p_{l,k} \mid l\in L, k\in K_l \}
	  }
      }
      {f(\zeta_1,\ldots,\zeta_{\rank(f)}) \trans[a] \theta }
  \]
  with ${\gtgeq_{l,k}} \in \{>, \geq\}$ for all $l\in L$ and $k\in K_l$, and
  $\mathcal{Z} = \diag{Y_l}_{l\in L} \times\prod_{\zeta \in W}\{\zeta\}$, with 
  $W\subseteq \TVar \cup \DVar \backslash \bigcup_{l\in L} Y_l$, 
 In addition, it has to satisfy the following conditions:
  \begin{compactenum}
  \item\label{item:conditions_on_cardinality}%
    Each set $Y_l$ should be at least countably infinite, for all
    $l\in L$, and the cardinality of $L$ should be strictly smaller
    than that of the $Y_l$'s.
  \item\label{item:conditions_on_z}%
    All variables $\zeta_1,\ldots,\zeta_{\rank(f)}$ are different.
  \item\label{item:condition_mus_are_different}%
    All variables $\mu_m^{\vec{z}}$, with $m\in M$ and
    $\vec{z}\in \mathcal{Z}$, are different and
    $\{\zeta_1,\ldots,\zeta_{\rank(f)}\} \cap
     \{\mu^{\vec{z}}_m \mid \vec{z} \in \mathcal{Z}, m\in M \}
     = \emptyset$.
  \item\label{item:conditions_on_Y}%
    For all $l\in L$,
    $Y_l \cap \{\zeta_1,\ldots,\zeta_{\rank(f)}\} = \emptyset$, and
    $Y_l \cap Y_{l'} = \emptyset$ for all $l'\in L$, $l\neq l'$.
  \item\label{item:conditions_on_nonrepeating_z}%
    For all $m\in M$, the set
    $\{{\mu^{\vec{z}}_m} \mid {\vec{z}\in \mathcal{Z}}\} \cap
     \left(\Var(\theta) \cup (\bigcup_{l \in L} \Var(\theta_l))\cup W\right)$
    is finite.
  \item\label{item:conditions_on_variables_of_theta}%
    For all $l\in L$, the set
    $Y_l \cap
         \left( \Var(\theta) \cup 
                \bigcup_{l'\in L}\Var(\theta_{l'})\right)$
    is finite.
  \end{compactenum}
  A rule $r\in R$ is in \emph{\ntmuxt{} format} if its form is like
  above but has a conclusion of the form $x \trans[a] \theta$ and, in
  addition, it satisfies the same conditions as above, except that
  whenever we write $\{\zeta_1,\ldots,\zeta_{\rank(f)}\}$, we should
  write $\{x\}$.
  $P$ is in \emph{\ntmuft{} format} if all its rules are in \ntmuft{}
  format.
  $P$ is in \emph{\ntmufxt{} format} if all its rules are in either 
  \ntmuft{} format or \ntmuxt{} format.
\end{definition}

The rationale behind each of the restrictions are discussed
in~\cite{DGL15a} in depth (see also~\cite{DL12}).
In the following we briefly summarize it.
Variables $\zeta_1,\ldots,\zeta_{\rank(f)}$ in the source of the
conclusion, all variables $\mu_m^{\vec{z}}$ in the target of the
positive premises, and all variables in the sets $Y_l$, $l\in L$, as
part of the measurable sets in the quantitative premises, are binding.
That is why all of them are requested to be different, which is stated
in conditions~\ref{item:conditions_on_z},
\ref{item:condition_mus_are_different},
and~\ref{item:conditions_on_Y}.
If $Y_l$ is finite, quantitative premises will allow to count the
minimum number of terms that gather certain probabilities.  This goes
against the spirit of bisimulation that measures equivalence classes of
terms regardless of the size of them.  Therefore $Y_l$ needs to be
infinite (condition~\ref{item:conditions_on_cardinality}).
Condition~\ref{item:conditions_on_nonrepeating_z} ensures that, for
each $m\in M$ there are sufficiently many distribution variables in the set
$\{\mu_m^{\vec{z}}\mid {\vec{z}\in\mathcal{Z}}\}$
to be freely instantiated.  The use of a distribution variable in a
quantitative premise may disclose part of the structural nature of the
distribution term that substitutes such variable.  Thus, for instance,
if all variables $\mu_m^{\vec{z}}$ are used in different quantitative
premises together with some lookahead, we may restrict the syntactic
form of the eventually substituted distribution terms, hence revealing
unwanted differences.  A similar situation arises with the use of
variables in $Y_l$ for all $l\in L$, hence
condition~\ref{item:conditions_on_variables_of_theta}.
The precise understanding of
conditions~\ref{item:conditions_on_nonrepeating_z}
and~\ref{item:conditions_on_variables_of_theta} requires a rather
lengthy explanation that is beyond the scope of this paper.  The
reader is referred to~\cite{DL12,DGL15a} for details.

All congruence theorems in this article apply only to PTSSs whose
rules are \emph{well-founded}.
A rule $r$ is \emph{well-founded} if there is no infinite backward
chain in the \emph{dependency graph} $G_r = (V, E)$ of $r$ defined by
$V = \TVar \cup \DVar$ and
$E =   \{\tuple{\zeta,\mu} \mid (t\trans[a]\mu)\in\pprem{r}, \zeta\in\Var(t)\}
  \cup \{\tuple{\zeta,y} \mid (\theta(Y)\gtgeq p)\in\qprem{r},
                              \zeta\in\Var(\theta), y\in Y\}$.
A PTSS is called \emph{well-founded} if all its rules are
well-founded.

The full proof of the following theorem can be found in~\cite{DGL15a}.

\begin{theorem}\label{th:bisimulation_congruence}
  Let $P = (\Sigma, A, R)$ be a complete well-founded PTSS in \ntmufxt\ format.
  Then, the bisimulation equivalence is a congruence for all operators
  defined in $P$.
\end{theorem}

The \ntmufxt\ format is still too general to preserve the other
(weaker) bisimulation equivalences presented in
Sec.~\ref{sec:semantic_relations}.  In the reminder of the section, we
will discuss through appropriate examples how the \ntmufxt\ format
should be further restricted or modified so that the other
bisimulation equivalences are congruences for the resulting restricted
formats.

We focus first on convex bisimulation. For this consider the terms
$\te_1 = a.(\dist{b.\nullproc}) + a.(\dist{c.\nullproc})$ and
$\te_2 = \te_1 + a.(\dist{b.\nullproc} \oplus_{0.5} \dist{c.\nullproc})$.
Notice that $\te_1 \bisim_c \te_2$.
Consider a possible extension of our running example with a unary
operator $f$ with the following \ntmuft\ rule:
{\small%
\begin{equation}\label{eq:convex:ex1}
  \dedrule{x\trans[a]\mu \quad
           \mu(Y)\geq 0.5 \quad \{y\trans[b]\mu_y \mid y \in Y\} \quad
           \mu(Y')\geq 0.5 \quad \{y'\trans[c]\mu_{y'} \mid y' \in Y'\}}
          {f(x)\trans[a]\dist{\nullproc}}
\end{equation}%
}%
Since
$\te_2 \trans[a] (\dist{b.\nullproc}\oplus_{0.5}\dist{c.\nullproc})$,
$f(\te_2)\trans[a]\dist{\nullproc}$.  However it is easy to see that
$f(\te_1)$ cannot perform any transition.
Therefore $f(\te_1) \not\bisim_c f(\te_2)$.

The problem arises precisely because, in order to show that $\te_1
\bisim_c \te_2$, transition
$\te_2 \trans[a] (\dist{b.\nullproc}\oplus_{0.5}\dist{c.\nullproc})$
is matched with the appropriate convex combination of the transitions
$\te_1\trans[a]\dist{b.\nullproc}$ and $\te_1\trans[a]\dist{c.\nullproc}$.
Thus, we need that a quantitative premise guarantees that the test is
produced on a convex combination of target distributions rather than
on a single target distribution.
An appropriate modification of such rule would be to replace it by a
family of rules of the form
{\small%
\[\dedrule{\{x\trans[a]\mu_n \mid n{\in}\N \} \ \ \,
           \left(\bigoplus_{n{\in}\N}[p_n]\mu_n\right)(Y)\geq 0.5 \ \ \, \{y\trans[b]\mu_y \mid y {\in} Y\} \ \ \,
           \left(\bigoplus_{n{\in}\N}[p_n]\mu_n\right)(Y')\geq 0.5 \ \ \, \{y'\trans[c]\mu_{y'} \mid y' {\in} Y'\}}
          {f(x)\trans[a]\dist{\nullproc}},
\]%
}%
one for each $\{p_n\}_{n\in\N}$ such that $\sum_{n\in\N}p_n = 1$ and
each $p_i\in[0,1]\cap\Q$.

\remarkMDL{This problem is related to the fact that d-sort positions 
do not distribute wrt. $\oplus$. See the last paragraph of Section 2.}
Consider now that the semantic of $f$ is defined by the rule
{\small%
\begin{equation}\label{eq:convex:ex2}
  \dedrule{x \trans[a] \mu}{f(x) \trans[a] \dist{a.}\mu}
\end{equation}%
}%
and notice that
$f(\te_2)\trans[a]\dist{a.}(\dist{b.\nullproc}\oplus_{0.5}\dist{c.\nullproc})$.
However, the only two possible transitions for $f(\te_1)$ are
$f(\te_1)\trans[a]\dist{a.b.\nullproc}$ and
$f(\te_1)\trans[a]\dist{a.c.\nullproc}$, and there is no $p\in[0,1]$
such that $\dist{a.b.\nullproc}\oplus_p\dist{a.c.\nullproc} \bisim_c
\dist{a.}(\dist{b.\nullproc}\oplus_{0.5}\dist{c.\nullproc})$.
For this reason, we will require that a target of a positive premise
does not appear in a \DTS-sorted position of a subterm in the target of
the conclusion.
\remarkPRD{%
  Couldn't we have considered, instead, the family of rules
  $\displaystyle
  \dedrule{\{x\trans[a]\mu_n \mid n{\in}\N \}}
  {f(x) \trans[a] \dist{a.}\left(\bigoplus_{n{\in}\N}[p_n]\mu_n\right)}$ ?}

For the next example, we consider an additional unary \DTS-sorted
operator $g$ and the following rules
{\small%
\begin{equation}\label{eq:convex:ex3}
  \dedrule{x\trans[a]\mu \qquad g(\mu)\trans[b]\mu'}
          {f(x)\trans[a]\dist{\nullproc}}
  \qquad\qquad
  \dedrule{\mu(Y)>0 \quad \{y\trans[b]\mu \mid y\in Y\} \quad
           \mu(Y')>0 \quad \{y'\trans[c]\mu' \mid y'\in Y'\}}
          {g(\mu) \trans[b] \dist{\nullproc}}
\end{equation}%
}%
Notice that
$g(\dist{b.\nullproc}\oplus_{0.5}\dist{c.\nullproc})\trans[b]\dist{\nullproc}$.  Therefore $f(\te_2)\trans[a]\dist{\nullproc}$.
However, neither $g(\dist{b.\nullproc})$ nor 
$g(\dist{c.\nullproc})$ can perform any transition, and as a consequence 
$f(\te_1)$ cannot perform any transition either.
Hence, $f(\te_1) \not\bisim_c f(\te_2)$.
For this reason we will require that a target of a positive premise
does not appear in the source of a positive or negative premise.
\remarkPRD{%
  Similarly as before, why not consider the family of rules
  $\displaystyle
  \dedrule{\{x\trans[a]\mu_n \mid n{\in}\N \} \quad g\left(\bigoplus_{n{\in}\N}[p_n]\mu_n\right)\trans[b]\mu'}
          {f(x)\trans[a]\dist{\nullproc}}$
  instead?}

Suppose now that $g$ is a binary \TS-sorted operator and consider the
following rules
{\small%
\begin{equation}\label{eq:convex:ex4}
  \dedrule{x\trans[a]\mu}{f(x)\trans[a]\dist{g}(\mu,\mu)}
  \qquad\qquad
  \dedrule{x_1 \trans[b] \mu_1 \qquad x_2 \trans[c] \mu_2}
          {g(x_1,x_2) \trans[a] \dist{\nullproc}}
\end{equation}%
}%
Notice that the only possible transitions for $f(\te_1)$ are
$f(\te_1)\trans[a]\dist{g}(\dist{b.\nullproc},\dist{b.\nullproc})$ and
$f(\te_1)\trans[a]\dist{g}(\dist{c.\nullproc},\dist{c.\nullproc})$.
Moreover, notice that
$\dist{g}(\dist{b.\nullproc},\dist{b.\nullproc}) \bisim_c
\dist{g}(\dist{c.\nullproc},\dist{c.\nullproc}) \bisim_c
\dist{\nullproc}$.
However,
$f(\te_2)\trans[a]\dist{g}(\dist{b.\nullproc}\oplus_{0.5}\dist{c.\nullproc},\dist{b.\nullproc}\oplus_{0.5}\dist{c.\nullproc})$,
and it is not difficult to see that
$\dist{g}(\dist{b.\nullproc}\oplus_{0.5}\dist{c.\nullproc},\dist{b.\nullproc}\oplus_{0.5}\dist{c.\nullproc})
\bisim_c (\dist{a.\nullproc}\oplus_{0.25}\dist{\nullproc})$.
Therefore, $f(\te_1) \not\bisim_c f(\te_2)$.
In this case, the problem seems to arise because the same distribution
variable occurs in the target of the conclusion of the first rule in
two different \TS-sorts positions of the target distribution term.
However, the problem is not so general.  Notice that if the target in
the conclusion is replaced by the term
$\dist{g}(\mu,\dist{c.\nullproc})\oplus_p\dist{g}(\dist{b.\nullproc},\mu)$
we would have $f(\te_1) \bisim_c f(\te_2)$.
The difference arises from the fact that in the interpretation of $\dist{g}(\theta,\theta)$ the probability distribution
$\sem{\theta}{}$ multiplies with itself.  This is not the case in the
interpretation of
$\dist{g}(\theta,\dist{c.\nullproc})\oplus_p\dist{g}(\dist{b.\nullproc},\theta)$
where the two instances of $\sem{\theta}{}$ are summed up.
Thus, we will actually request that the target of the conclusion is
\emph{linear} with respect to each distribution variable on a target
of a positive premise.

\begin{definition}
  \label{def:linear_term}
  A distribution term $\theta \in \openDTerms$ is \emph{linear for a
    set $V\subseteq\DVar$} if
  \begin{inparaenum}[(i)]
  \item%
    $\theta \in \closedDTerms \cup \DVar \cup \{\delta(x) \mid x \in \TVar\}$.
  \item%
    $\theta = \bigoplus_{i \in I} [p_i]\theta_i$ and $\theta_i$ is
    linear for $V$, for all $i\in I$,
  \item%
    $\theta = f(\theta_1, \ldots, \theta_n)$,
    for all $i\in I$, $\theta_i$ is linear for $V$, and
    $\Var(\theta_i)\cap\Var(\theta_j)\cap V=\emptyset$,
    for all $i,j\in\{1,\ldots,n\}$ and $i \neq j$,
  \end{inparaenum}
\end{definition}

\begin{definition}\label{def:convex_format}
  Let $P = (\Sigma, A, R)$ be a PTSS.  A rule $r\in R$ is in
  \emph{convex \ntmuft\ format} if has the form
  \[\ddedrule{%
    \begin{array}{c}
      {\textstyle
	\bigcup_{m\in M}
        \{ t_m(\vec{z})\trans[a_m]\mu_m^{\vec{z}} \mid \vec{z}\in \mathcal{Z} \} \qquad
        \bigcup_{n\in N}
        \{ t_n(\vec{z})\ntrans[b_n] \mid \vec{z}\in \mathcal{Z} \} 
      } \\[1ex]
      {\textstyle
	\bigcup_{\tm \in \tM} \{t_{\tm}(\vec{z}_{\tm}) \trans[a_\tm] \mu^\tm_i \mid i \in\N\} \qquad
	\bigcup_{\tm \in \tM} \{\big(\bigoplus_{i\in\N} [p^\tm_i] \mu^\tm_i\big) (Y_{l}) \gtgeq_{l,k} p_{l,k} \mid l \in L_\tm, k \in K_l\}
      } 
    \end{array}
  }%
  {f(\zeta_1,\ldots,\zeta_{\rank(f)}) \trans[a] \theta }
  \]
  with
  $L = \cup_{\tm \in \tM} L_\tm$, $L_\tm\cap L_{\tm'} = \emptyset$
  whenever $\tm\neq\tm'$,
  ${\gtgeq_{l,k}} \in \{>, \geq\}$ for all $l\in L$ and $k\in K_l$,
  $\mathcal{Z} = \diag{Y_l}_{l\in L} \times \prod_{\zeta \in W}\{\zeta\}$, with 
  $W\subseteq \TVar \cup \DVar \backslash \bigcup_{l\in L} Y_l$. 
  In addition, it should also satisfy
  conditions~\ref{item:conditions_on_cardinality}
  to~\ref{item:conditions_on_variables_of_theta} in
  Def.~\ref{def:ntmufxtheta} and the following extra conditions:
  \begin{compactenum}
    \setcounter{enumi}{6}%
  \item\label{item:convex:condition_on_scalars}%
    For every $\tm\in\tM$, the family
    $\{p^\tm_i\}_{i\in\N} \subseteq {[0,1]\cap\Q}$ and
    $\sum_{i \in\N} p^\tm_i = 1$
  \item\label{item:convex:condition_on_auxiliary_pp}%
    For every $\tm \in \tM$, there is exactly one $j\in\N$ such that
    $\mu^\tm_j = \mu^{\vec{z}}_m$  for some $m \in M$ and
    $\vec{z} \in \mathcal{Z}$, in which case also
    $t_\tm(\vec{z}_\tm)\trans[a_\tm]\mu^\tm_j =
    t_m(\vec{z})\trans[a_m]\mu^{\vec{z}}_m$.
    Moreover
    $\{\mu^\tm_i\mid i\in\N\} \cap \{\mu^{\tm'}_i\mid i\in\N\} =
    \emptyset$ for all $\tm\neq\tm'$, and
    $\{\mu^\tm_i\mid i\in\N\} \cap \{\zeta_1,\ldots,\zeta_{\rank(f)}\}
    = \emptyset$.
  \item\label{item:convex:condition_on_limited_use_of_mu}%
    No variable $\mu_m^{\vec{z}}$, with $m\in M$ and
    $\vec{z}\in\mathcal{Z}$, appears in the source of a premise
    (i.e.\ in the set $W$) or in a \DTS-sorted position of a subterm in
    the target of the conclusion $\theta$.
  \item\label{item:convex:condition_on_target}%
    $\theta$ is linear for
    $\{\mu^{\vec{z}}_m \mid m \in M, \vec{z} \in \mathcal{Z} \}$.
  \end{compactenum}
  A rule $r\in R$ is in \emph{convex \ntmuxt{} format} if its form is
  like above but has a conclusion of the form $x \trans[a] \theta$ and
  it satisfies the same conditions, except that whenever we write
  $\{\zeta_1,\ldots,\zeta_{\rank(f)}\}$, we should write $\{x\}$.
  A set of convex \ntmufxt\ rules $R$ is \emph{convex closed} if for
  all $r\in R$, for any term $\bigoplus_{i\in\N} [p^\tm_i] \mu^\tm_i$
  appearing in a quantitative premise of $r$ and any family
  $\{q_i\}_{i\in\N} \subseteq {[0,1]\cap\Q}$ such that $\sum_{i \in\N}
  q_i = 1$, then the rule $r'$ obtained by replacing each occurrence
  of $\bigoplus_{i\in\N} [p^\tm_i] \mu^\tm_i$ in $r$ by
  $\bigoplus_{i\in\N} [q_i] \mu^\tm_i$ is also in $R$.
  A PTSS $P = (\Sigma, \Act, R)$ is in \emph{convex \ntmufxt{} format}
  if all rules in $R$ are in convex \ntmufxt{} format and $R$ is
  convex closed.
\end{definition}

The problem indicated in rule~(\ref{eq:convex:ex1}) is attacked
with the requirement of having sets
$\{t_{\tm}(\vec{z}_{\tm})\trans[a_\tm]\mu^\tm_i \mid i \in\N\}$ as
positive premises with which the convex closures $\bigoplus_{i\in\N}
[p^\tm_i] \mu^\tm_i$ can be constructed, plus the request that the set
of rules is convex closed.  Notice that
condition~\ref{item:convex:condition_on_auxiliary_pp} states that
these sets of positive premises are only used to construct such
distribution terms and are only linked to the ``actual'' positive
premises in
$\bigcup_{m\in M} \{ t_m(\vec{z})\trans[a_m]\mu_m^{\vec{z}} \mid
\vec{z}\in \mathcal{Z} \}$
through a single transition $t_{\tm}(\vec{z}_{\tm})\trans[a_\tm]\mu^\tm_j$.
\remarkPRD{Revise text related to
  condition~\ref{item:convex:condition_on_auxiliary_pp}}

Rules like~(\ref{eq:convex:ex2}) and the left rule on~(\ref{eq:convex:ex3}) are
excluded on condition~\ref{item:convex:condition_on_limited_use_of_mu}
since no variable of a positive premise can be used in the source of
a premise (excluding~(\ref{eq:convex:ex3})) or in a \DTS-sort position
in the target of the conclusion (excluding~(\ref{eq:convex:ex2})).
Finally, rules like on the left of~(\ref{eq:convex:ex4}) are excluded
by requesting that the target of the conclusion is linear
(condition~\ref{item:convex:condition_on_target}).

Now, we can state the congruence theorem for convex bisimulation
equivalence.

\begin{theorem}\label{th:convex_bisimulation_congruence}
  Let $P$ be a complete well-founded PTSS in convex \ntmufxt\ format.  Then,
  convex bisimulation equivalence is a congruence for all operators
  defined by $P$.
\end{theorem}

We focus now on the probability abstracted bisimulation.  Notice that
the terms
$\te_3 = a.(\dist{b.\nullproc} \oplus_{0.5} \dist{c.\nullproc})$ and
$\te_4 = a.(\dist{b.\nullproc} \oplus_{0.1} \dist{c.\nullproc})$
are probability abstracted bisimilar, i.e., $\te_3 \bisim_a \te_4$.
Consider now the unary operator $f$ whose semantics is defined with
rule~(\ref{eq:convex:ex1}).  It should not be difficult so see that 
$f(\te_3)\trans[a]\dist{\nullproc}$ while
$f(\te_4)$ cannot perform any transition.
Therefore $f(\te_3) \not\bisim_a f(\te_4)$.
The problem is a consequence of the fact that the quantitative premises are tested
against non-zero values which may distinguish distributions with the
same support set but mapping into different probability values.
Thus, in order to preserve probability abstracted bisimulation
equivalence, the only extra restriction that we ask to a rule in
\ntmufxt\ format is that none of its quantitative premises test
against a value different from 0.

\begin{definition}\label{def:abstract_format}
  A PTSS $P = \tuple{\Sigma,A,R}$ is in
  \emph{probability abstracted \ntmufxt\ format} if
  it is in \ntmufxt\ format and
  for every rule $r \in R$ and quantitative premise
  ${\theta(Y)\gtgeq p} \in \qprem{r}$, $p=0$.
\end{definition}

The proof of the congruence theorem for probability abstracted bisimulation equivalence (Theorem~\ref{th:prob_abstracted_format} below) follows closely the
lines of the proof of Theorem~\ref{th:bisimulation_congruence} as
given in \cite{DGL15a}.

\begin{theorem}\label{th:prob_abstracted_format}
  Let $P$ be a complete well-founded PTSS in probability abstracted
  \ntmufxt\ format.  Then, the probability abstracted bisimulation
  equivalence is a congruence for all operators defined in $P$.
\end{theorem}

Given the alternative definition of the probability obliterated
bisimulation provided by Lemma~\ref{lemma:pob_using_otrans}, we will
now consider simpler definitions for the quantitative premises for the
rule format associated to this relation.  Thus, we consider quantitative
premises of the form $\theta(\{y\})\gtgeq p$ rather than
$\theta(Y)\gtgeq p$.

Taking $\te_3$ and $\te_4$ as before, we have that $\te_3 \bisim_o \te_4$.
The same example of the unary operator $f$, whose semantics is defined
with a conveniently modify rule~(\ref{eq:convex:ex1}), shows that
$f(\te_3) \not\bisim_a f(\te_4)$ and hence the need that the quantitative
premises can only be tested against 0.

Let
$\te_5 = a.(\dist{b.\nullproc} \oplus_{0.5} \dist{c.\nullproc})$ and
$\te_6 = a.\dist{b.\nullproc} + a.\dist{c.\nullproc}$, and observe that
$\te_5 \bisim_o \te_6$.
Take rule~(\ref{eq:convex:ex2}) as the semantic definition for $f$.
Notice that
$f(\te_5)\trans[a]\dist{a.}(\dist{b.\nullproc} \oplus_{0.5} \dist{c.\nullproc})$
is the only transition for $f(\te_5)$, while the only
possible transitions for $f(\te_6)$ are
$f(\te_6)\trans[a]\dist{a.}\dist{b.\nullproc}$ and
$f(\te_6)\trans[a]\dist{a.}\dist{c.\nullproc}$.
Since
${a.\dist{b.\nullproc}} \not\bisim_o
{a.(\dist{b.\nullproc}\oplus_{0.5}\dist{c.\nullproc})} \not\bisim_o
{a.\dist{c.\nullproc}}$,
$f(\te_5) \not\bisim_o f(\te_6)$.
Like for the convex bisimulation case, this shows that the target of a
positive premise cannot appear in a \DTS-sorted position of a subterm
in the target of the conclusion.

Suppose now that the semantics of $f$ is defined with the rule
{\small%
\begin{equation}\label{eq:oblit:ex1}%
  \dedrule{x \trans[a] \mu \qquad \mu(\{y_1\}) > 0 \qquad \mu(\{y_2\}) > 0
           \qquad y_1 \trans[b] \mu_1 \qquad y_2 \trans[c] \mu_2}%
          {f(x) \trans[a] \dist{\nullproc}}
\end{equation}%
}%
Notice that $f(\te_5) \not\bisim_o f(\te_6)$ since
$f(\te_5)\trans[a]\dist{\nullproc}$ while $f(\te_6)$ cannot perform any
transition.
This is due to the fact that, by allowing the same distribution
variable $\mu$ to occur in different quantitative premises, we gain
some knowledge of the structure of (the instance of) $\mu$, in
particular of its support set.

Consider now that $f$ is defined with the left rule
in~(\ref{eq:convex:ex3}) and $g$ with an appropriate modification of
the right rule in~(\ref{eq:convex:ex3}).
Notice that $f(\te_5)\trans[a]\dist{\nullproc}$ but $f(\te_6)$ cannot
perform any transition.  Thus $f(\te_5) \not\bisim_o f(\te_6)$.  In this
case, we are also gaining knowledge of the support set of $\mu$, but
this time through the rule associated to the operator $g$.  Therefore
we require that a target of a positive premise does not appear in the
source of a positive or negative premise.

Consider now the rules
{\small%
\begin{equation}\label{eq:oblit:ex2}%
  \dedrule{x \trans[a] \mu \quad \mu(\{y\})>0 \quad y\trans[b]\mu'}%
          {f(x) \trans[a] \dist{g}(\mu)}
  \qquad
  \dedrule{x \trans[c] \mu}%
          {g(x) \trans[c] \dist{\nullproc}}
\end{equation}%
}%
Notice that the only transition for $f(\te_5)$ is
$f(\te_5)\trans[a]\dist{g}(\dist{b.\nullproc} \oplus_{0.5} \dist{c.\nullproc})$
and the only transition for $f(\te_6)$ is
$f(\te_6)\trans[a]\dist{g}(\dist{b.\nullproc})$.
Then $f(\te_5)\otrans[a]g(c.\dist{\nullproc})\otrans[c]{\nullproc}$ while
$f(\te_6)\otrans[a]g(b.\dist{\nullproc})$ is the only possible
``obliterated'' transition for $f(\te_6)$.
Then $f(\te_5) \not\bisim_o f(\te_6)$.
This is an alternative way of gaining information on the support set
of a possible instance of $\mu$ in~(\ref{eq:oblit:ex2}): on the one
hand, by the quantitative premise on the first rule, we deduce that
such instance has an element in the support set that performs a
$b$-transition and, on the other hand, by having $\mu$ as an argument
in the target of the conclusion, we may gather extra information from
the same instance of $\mu$ through the rules for the semantics of the
target of the conclusion (in this case, that $\mu$ has another element
in the support set that performs a $c$-transition.)
Therefore, we forbid that the target of a positive premise is
both tested in a quantitative premise and used in the target of the
conclusion.

Notice that the example in rules~(\ref{eq:convex:ex4}) also apply for
probability obliterated bisimulation since $\te_1\bisim_o \te_2$ but
$f(\te_1) \not\bisim_o f(\te_2)$ with exactly the same explanation.  Thus,
we also request that the target of the conclusion is linear for all
distribution variables on targets of positive premises.

Finally, consider a modification~(\ref{eq:convex:ex4}) where the left
rule is instead
{\small%
\begin{equation}\label{eq:oblit:ex3}%
  \dedrule{x \trans[a] \mu \quad \dist{g}(\mu,\mu)(\{y\})>0 \quad y\trans[a]\mu'}%
          {f(x) \trans[a] \dist{\nullproc}}
\end{equation}%
}%
It should not be difficult to observe that
$f(\te_5)\trans[a]\dist{\nullproc}$ but $f(\te_6)$ cannot perform any
transition.  Thus $f(\te_5) \not\bisim_o f(\te_6)$.  For this reason we
also require that the quantitative premises only allow linear
distribution terms.

\begin{definition}\label{def:prob_ob_format}
  Let $P = (\Sigma, \Act, R)$ be a well-founded PTSS.
  A rule $r\in R$ is in \emph{probability obliterated \ntmuft\ format} 
  if it has the form %
  \[
  \ddedrule%
      {
	{\textstyle
          \bigcup_{m\in M}
            \{ t_m \trans[a_m]\mu_m \} \quad \cup \quad
          \bigcup_{n\in N}
            \{ t_n \ntrans[b_n] \} \quad  \cup \quad
	  \bigcup_{l \in L} \{ \theta_l(\{y_l\}) >  0\}
	}
      }
      {f(\zeta_1,\ldots,\zeta_{\rank(f)}) \trans[a] \theta }
  \]
  where all variables $\zeta_1, \ldots, \zeta_{\rank(f)}$, $\mu_m$, with $m\in
  M$, and $y_l$, with $l\in L$, are different and the following
  restrictions are satisfied:
  \begin{compactenum}
  \item\label{item:oblivion:condition_on_limited_use_of_mu}%
    For all $m \in M$,
    $\Var(t_m)\cap\{\mu_{m'}\mid m'\in M \}=\emptyset$.
    Similarly, for all $n \in N$,
    $\Var(t_n)\cap\{\mu_{m'}\mid m'\in M \}=\emptyset$.
  \item\label{item:oblivion:condition_on_qprem}%

    For all $l \in L$, $\theta_l$ is linear for $\{\mu_{m'}\mid m'\in M \}$
    and, moreover, for all $l, l' \in L$ with $l \neq l'$,
    $\Var(\theta_{l}) \cap \Var(\theta_{l'}) \cap \{\mu_m \mid m \in M\}
    = \emptyset$.
  \item\label{item:oblivion:condition_on_target}%
    $\theta$ is linear for $\{\mu_{m'}\mid m'\in M \}$,
    $\Var(\theta) \cap \big(\bigcup_{l\in L}\Var(\theta_l)\big)
     \cap \{\mu_m \mid m \in M\} = \emptyset$,
    and no variable $\mu_m$ appear in a \DTS-sorted position of a
    subterm of the target of the conclusion $\theta$.
  \end{compactenum}
  A rule is in \emph{probability obliterated \ntmuxt{} format} if its
  form is like above but has a conclusion of the form $x \trans[a]
  \theta$.  $P$ is in \emph{probability obliterated \ntmufxt{} format}
  if all its rules are in probability obliterated \ntmufxt{} format.
\end{definition}

Condition~\ref{item:oblivion:condition_on_limited_use_of_mu} limits
the form to exclude rules like the one on the left
of~(\ref{eq:convex:ex3}).
Condition~\ref{item:oblivion:condition_on_qprem} requires that the
distribution terms on the quantitative premises are linear
(excluding~(\ref{eq:oblit:ex2})), and that they do not share
distributions variables on the target of positive premises
(excluding~(\ref{eq:oblit:ex1})).
Finally, condition~\ref{item:oblivion:condition_on_target} request
that the target of the conclusion is linear
(excluding~(\ref{eq:convex:ex4})) and does not have targets of
positive premises on \DTS-sorted positions
(excluding~(\ref{eq:convex:ex2})) nor if they are used in quantitative
premises (excluding~(\ref{eq:oblit:ex2})).

Finally, we state the congruence theorem for probability obliterated
bisimulation equivalence.

\begin{theorem}\label{th:prob_obliterated_format}
  Let $P$ be a complete well-founded PTSS in probability obliterated \ntmufxt{}
  format.  Then, probability obliterated bisimulation equivalence is a
  congruence for all operators defined by $P$.
\end{theorem}


\section{Conclusion and Future Work}\label{sec:conclusion}

In this article, we presented three new rule formats that preserve
three different bisimulation equivalences coarser than Larsen \&
Skou's bisimulation.
These formats are more restricted variants of the \ntmufxt\ format and
notably, all of them can be seen as generalizations of the
non-probabilistic \emph{ntyft/ntyxt} format~\cite{Gro93,BG96}.
For completeness we mention two other similar results on PTSSs that
fall out of Larsen \& Skou's bisimulation.  They are~\cite{LV15}, that
presents a format for rooted branching bisimulation, and~\cite{Tin10},
that presents a format for non-expansiveness of
$\epsilon$-bisimulations.

Prior to the congruence theorems, we presented the different
bisimulation equivalences, compare them, and, in particular, we
gave a logic characterization for each of them.
The intention of presenting these logic characterizations is to use
them as the basis for the proof of full abstraction theorems~(see,
e.g., \cite{GV92,Gro93,DL12,DGL15a}.)
Full abstraction theorems are somewhat dual to the congruence theorems.
An equivalence relation is fully abstract with respect to a particular
format and an equivalence relation $\equiv$ if it is the largest
relation included in $\equiv$ that is a congruence for all operators
definable by any PTSS in that format.
In particular we are interested when $\equiv$ is the coarsest
reasonable behavioral equivalence, namely, (possibilistic) trace
equivalence.
We are busy now on trying to prove this results for the formats
presented here using the logic characterization as a means to
construct the so called \emph{testers}.
As the current point of our investigation, we do not
foresee major problems for all relations except for convex bisimulation
equivalences, for which we may need to relax some of the conditions of
the convex \ntmufxt\ format.

\bibliographystyle{eptcs}

\end{document}